\documentclass[aps,pra,twocolumn,nofootinbib]{revtex4-1}
\usepackage{amsfonts,amssymb,amsmath}            % for math symbols.
\usepackage{lmodern}
\usepackage[]{graphics,graphicx,epsfig}            % for graphics figures.
\usepackage{amsthm,bbold}
\usepackage{hyperref}
\newtheorem{theorem}{Theorem}

\newtheorem{lemma}{Lemma}
\newtheorem{problem}{Problem}
\newcommand{\ket}[1]{\left | #1 \right\rangle}
\newcommand{\bra}[1]{\left \langle #1 \right |}
\newcommand{\half}{\mbox{$\textstyle \frac{1}{2}$}}

\newcommand{\Tr}{\text{Tr}}

\newcommand{\proj}[1]{\ket{#1}\bra{#1}}
\newcommand{\identity}{\mathbb{1}}
\renewcommand{\epsilon}{\varepsilon}
\input{Qcircuit}
\begin{document}

\title{The Degree of Quantum Correlation Required to Speed-Up a Computation}
\date{\today}

\author{Alastair \surname{Kay}}
\affiliation{Department of Mathematics, Royal Holloway University of London, Egham, Surrey, TW20 0EX, UK}
%\author{Leong-Chuan \surname{Kwek}}
%\affiliation{Centre for Quantum Technologies, National University of Singapore, Singapore 117543}
%\affiliation{Institute of Advanced Studies (IAS) and National Institute of Education, Nanyang Technological University, Singapore 639673}
\begin{abstract}
The one clean qubit model of quantum computation (DQC1) efficiently implements a computational task that is not known to have a classical alternative. During the computation, there is never more than a small but finite amount of entanglement present, and it is typically vanishingly small in the system size. In this paper, we demonstrate that there is nothing unexpected hidden within the DQC1 model -- Grover's Search, when acting on a mixed state, {\em provably} exhibits a speed-up over classical with guarantees as to the presence of only vanishingly small amounts of quantum correlations (entanglement and quantum discord) -- while arguing that this is not an artefact of the oracle-based construction. We also present some important refinements in the evaluation of how much entanglement may be present in DQC1, and how the typical entanglement of the system must be evaluated.
\end{abstract}

\maketitle

\section{Introduction}

Any computation whose quantum algorithm has a superior scaling of running time compared to its classical counterpart ought to pass through an intermediate state with non-trivial quantum properties, e.g.\ entanglement. Indeed, for pure-state computations, it has been shown that entanglement is necessary \cite{ent}. However, computation involving mixed states is a far more subtle issue, with no firm resolution, although it is looking increasingly likely that quantum discord must be present in order for there to be an exponential speed-up \cite{eastin2010,cable2015}.

The DQC1 model \cite{original} can provide important insights. Expressed as a decision problem, this is defined as:
\begin{problem}
Given an efficient classical description of a quantum computation $U$ on $n$ qubits, and a promise that either $\Tr(U)>1/\text{poly}(n)$ or $\Tr(U)<-1/\text{poly}(n)$, determine which is the case with error probability $\varepsilon<\frac{1}{3}$.
\end{problem}
\noindent
There is no known classical algorithm which can efficiently solve this problem, while there is a quantum circuit that is remarkably simple (Fig.\ \ref{fig:dqc1}). In this circuit, there is one special (`clean') qubit which is initially prepared in the state $\ket{+}=(\ket{0}+\ket{1})/\sqrt{2}$, and a set of $n$ qubits are prepared in the maximally mixed state $\identity/2^n$. After applying controlled-$U$ between the clean qubit (control) and the mixed qubits (target), we can estimate the probabilities that the clean qubit is in the state $\ket{+}$ (to find $\text{Re}(\Tr(U))$) or $(\ket{0}+i\ket{1})/\sqrt{2}$ (to find $\text{Im}(\Tr(U))$). There is a very obvious division in the system between the clean qubit and the mixed ones. As such, the authors of \cite{original} examined the entanglement between that bipartition of the system, and discovered it was 0. Entanglement seemed to not be necessary for mixed state quantum computation, although other measures of non-classicality such as the quantum discord and measurement-induced disturbance have since been shown to be non-zero across that bipartition \cite{discord,luo}.

\begin{figure}[!tb]
$$
\Qcircuit @C=1em @R=.7em {
\lstick{\ket{+}}	& \ctrl{1} & \meter \\
\multiinput{2}{\frac{\identity}{2^n}} & \multigate{2}{U} & \qw \gategroup{2}{1}{4}{1}{.7em}{\{}	\\
 & \ghost{U} & \qw	\\
 & \ghost{U} & \qw	\\
}
$$
\caption{The DQC1 model. There is no entanglement present in the bipartite split between the top (special) qubit and the rest. Estimation of the trace of $U$ proceeds by performing measurements in the bases $(\ket{0}\pm\ket{1})/\sqrt{2}$ and $(\ket{0}\pm i\ket{1})/\sqrt{2}$.} \label{fig:dqc1}
\vspace{-0.5cm}
\end{figure}
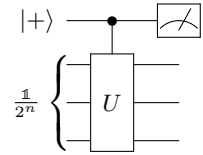

Detecting entanglement in a multipartite state using bipartite divisions is, however, quite subtle. A clear demonstration of this is the distribution of entanglement using separable states \cite{sepdist,kay3}; just because two of the three possible bipartitions of a 3-qubit system have no entanglement does not mean that the third bipartition has no entanglement. So it is with the DQC1 model -- in \cite{steve}, it was demonstrated that by considering different bipartitions, there is indeed entanglement present. Two additional results are provided in \cite{steve}: upper and lower bounds on the maximum amount of entanglement present across any bipartition and a numerical investigation of the entanglement produced by a typical unitary (in the sense of selecting a $U$ uniformly at random from the Haar measure). This investigation suggests that the entanglement is typically exponentially small in $n$. So, although there is some entanglement present, it is a vanishingly small amount on average.

Even the presence of {\em almost} no entanglement in the computation might be seem surprising. However, in this paper, we draw parallels with a noisy version of Grover's Search \cite{grover} in which we demonstrate the existence of a quantum speed-up in the presence of vanishingly small entanglement (such a result could have been proven from \cite{contain} but does not appear to have been), and other non-classicality measures such as the quantum discord \cite{discord1,discord2}. Moreover, unlike the case of DQC1 where the speed-up over classical is only {\em believed}, Grover's search is subject to an oracle-based complexity classification with a proven gap over the classical case. Placed in this context, the power of DQC1 seems much less surprising. However, we only demand the existence of a quantum speed-up, and do not address the important conceptual transition between a polynomial and an exponential speed-up. We also improve the results of \cite{steve}, giving a tight upper bound on the entanglement present in the DQC1 model for any $U$. 

A further criticism of the DQC1 model (or, similarly, the Grover search that we describe here) is that it assumes an implementation of the controlled-$U$ gate. In practice, we have no such implementation, and need to decompose the action in terms of elementary gates. Even if the net implementation of controlled-$U$ yields little or no entanglement, there is no {\em a priori} reason why all the intermediate steps should also be almost entanglement free across all bipartitions. On the other hand, if the controlled-$U$ is treated as an oracle, then we can consider other oracle based problems that are even more trivial in the entanglement sense -- the Bernstein-Vazirani problem transforms a product state into a product state \cite{BV,BV2}. The internal workings of the oracle are evidently important, and receive further discussion.

%In the Appendix, evaluation of the typical entanglement across a bipartition in the DQC1 model is examined. We describe a class of unitaries for which the entanglement is a constant, illustrating an important feature which \cite{steve} overlooked.

\subsection{Notation}

Consider a density matrix of $n+1$ qubits, $\rho$. Whenever we write $\rho$ in the context of DQC1, we refer specifically to the output of the circuit for estimating the trace of $U$ (Fig.\ \ref{fig:dqc1}), which is a block matrix of the form
$$
\rho=\frac{1}{2^{n+1}}\left(\begin{array}{cc}
\identity & U^\dagger \\ U & \identity
\end{array}\right).
$$
We denote by $\rho_y$, where $y\in\{0,1\}^n$, the partial transpose of $\rho$ taken over the non-clean qubits $i: y_i=1$, and never the clean qubit, without losing generality. Similarly, $U_y$ denotes the partial transpose of $U$, such that
$$
\rho_y=\frac{1}{2^{n+1}}\left(\begin{array}{cc}
\identity & U_y^\dagger \\ U_y & \identity
\end{array}\right).
$$
The Hamming weight of the string $y$ is denoted by $w_y$.

Throughout this paper, the multiplicative negativity \cite{mult1,mult2} is chosen as the entanglement measure in order to maintain consistency with \cite{steve}.
If $\Lambda_y=\text{spec}\left(\rho_y\right)$ is the spectrum of $\rho_y$, then the multiplicative negativity of $\rho$ is
$$
M_y=\sum_{\lambda\in\Lambda_y}|\lambda|
$$
for that specific bipartition. If no entanglement is present, $M_y=1$ for all $y\in\{0,1\}^n$.

\section{Maximum Entanglement in DQC1}

We start by examining the entanglement properties of the DQC1 model, strengthening the assertions in \cite{steve}.

\begin{theorem} \label{thm:1}
For any unitary $U$, the output $\rho$ of the DQC1 computation always satisfies
$$
M_y\leq \frac{5}{4}
$$
for all possible bipartitions $y\in\{0,1\}^n$.
\end{theorem}
\begin{proof}
Since it is not immediately clear that $U_y$ is normal, we use the singular value decomposition of $U_y$,
$$
U_y=RDV^\dagger
$$
where $R$ and $V$ are unitary matrices and $D$ is diagonal, with non-negative diagonal elements $d_i$. The eigenvalues of $\rho_y$ are unchanged under the action of unitary operations, so apply the transformation of controlled-$R^\dagger$ (controlled off the clean qubit), followed by a Pauli-$X$ on the clean qubit, controlled-$V^\dagger$, and finished by another Pauli-$X$ on the clean qubit. This yields
$$
\rho_y\mapsto\frac{1}{2^{n+1}}\left(\begin{array}{cc}
\identity & D \\ D & \identity \end{array}\right).
$$
The eigenvalues of $\rho_y$ are therefore readily calculated to be $(1\pm d_i)/2^{n+1}$ for $i=1\ldots 2^n$.

How, then, are we to pick $\{d_i\}$ such that the value of 
$$
M_y=1+2\sum_{i:d_i>1}\frac{d_i-1}{2^{n+1}}
$$
is maximised? We perform a constrained optimisation by noting from \cite{steve} that
$$
2^n=\Tr(UU^\dagger)=\Tr(U_yU_y^\dagger)=\sum_id_i^2.
$$
Assume that $t$ of the values $d_i>1$ ($i=1\ldots t$). Hence,
$$
M_y\leq\max_{t,\{d_i\}}1-\frac{t}{2^n}+\frac{1}{2^n}\sum_{i=1}^td_i
$$
subject to $\sum_{i=1}^{2^n}d_i^2=2^n$. The optimal choice of the $d_i$ is clear: $d_i=0$ for $i=t+1\ldots 2^n$ and
$$
d_i=\sqrt{\frac{2^n}{t}}
$$
otherwise, which gives
$$
M_y\leq\max_t1-\frac{t}{2^n}+\sqrt{\frac{t}{2^n}}.
$$
This is maximised by $t=2^n/4$ (requiring $n\geq 2$ such that $t$ is an integer), yielding
$$
M_y\leq\frac{5}{4},
$$
compared to the maximum possible value of $\approx 2^{n/2}$.
\end{proof}
The limit $M_y=\frac{5}{4}$ is readily saturated. For example, when $n=2$ we can use
$$
U=\left(\begin{array}{cccc}
0 & 0 & 0 & 1 \\
0 & 1 & 0 & 0 \\
0 & 0 & 1 & 0 \\
1 & 0 & 0 & 0
\end{array}\right)
$$
since $U_{01}$ has eigenvalues $2,0,0,0$ as required for the optimal construction. For larger $n$, a construction such as $U\otimes\identity^{\otimes(n-2)}$ would provide the requisite eigenvalues. There are then many ways of dressing this operator with unitaries that don't affect the partial transpose to disguise its structure slightly. For example, \cite{steve} uses a series of controlled-not gates. Incidentally, this trivial example shows that there are many bipartitions for which $M_y=\frac{5}{4}$: of the $2^n-1$ possible choices, all $2^{n-1}$ choices that have $y_1\oplus y_2=1$ exhibit this value.

This section reiterates the conclusion of \cite{steve}: there are bipartitions of the DQC1 model in which there is some entanglement present, and it can be present up to a finite amount. Indeed, there can be many such bipartitions. What we want to investigate in the rest of this paper is how surprising this result is -- should the fact that we are guaranteed that no more than a small but finite amount of entanglement is present in a bipartition suggest to us that a quantum speed-up should be more difficult to realise? What about the fact that for most unitaries it seems that the entanglement is vanishingly small (exponentially small in the system size, $N$)? To that end, we are now going to compare to a depolarised version of Grover's Search, for which we will show that even in situations where there is a provable (oracle-based) quantum speed-up, this can occur when we are guaranteed that there is never more than an exponentially small amount of entanglement in any bipartition. This is a much stronger statement than has been made for DQC1. As such, using DQC1 to make conclusions about the existence of a computational speed-up in the presence of little or no entanglement seems obsolete.

\section{Comparison with Grover's Search}

A standard formulation of Grover's Search algorithm \cite{grover} starts with an initial state $\ket{+}^{\otimes n}$, trying to evolve to some (unknown) target state $\ket{x_0}$. The evolution is restricted to a subspace spanned by $\ket{x_0}$ and
$$
\ket{\psi^\perp}=\frac{1}{\sqrt{2^n-1}}\sum_{x\neq x_0}\ket{x},
$$
such that the populated states are of the form
$$
\ket{\Psi}=\cos\theta\ket{x_0}+\sin\theta\ket{\psi^\perp}
$$
where $\theta\in[0,\pi/2]$. For finite $n$, only integer multiples of $\sin^{-1}(2^{-n})$ are allowed. If we select a particular bipartition, then $\ket{x_0}=\ket{x_y}\ket{x_{\bar y}}$ where, here, the subscript $y$ merely refers to the set of qubits on which the state is defined. Similarly, we can define
$$
\ket{\psi^\perp_y}=\frac{1}{\sqrt{2^{w_y}-1}}\sum_{x\in\{0,1\}^{w_y}:x\neq x_y}\ket{x},
$$
allowing $\ket{\psi^\perp}$ to be expressed as
\begin{equation*}
\begin{split}
\ket{\psi^\perp}=&\sqrt{\frac{(2^{w_y}-1)(2^{n-w_y}-1)}{2^n-1}}\ket{\psi^\perp_y}\ket{\psi^\perp_{\bar y}} \\
&\!\!\!\!+\sqrt{\frac{2^{n-w_y}-1}{2^n-1}}\ket{x_y}\ket{\psi^\perp_{\bar y}} +\sqrt{\frac{2^{w_y}-1}{2^n-1}}\ket{\psi^\perp_y}\ket{x_{\bar y}}.
\end{split}
\end{equation*}
Evidently, the state across this bipartition may be written as a state with no more than two Schmidt coefficients. The maximum entanglement is when the two Schmidt coefficients are $\lambda_1=\lambda_2=\half$ ($w_y=1$ as $n\rightarrow\infty$), yielding a value of $M_y=2$ for this standard, pure state, version of Grover's search. There is only ever a finite amount of entanglement between any bipartition, much like DQC1, the only difference being the value.

We will now show that there is a variant of this algorithm for which the entanglement can be made vanishingly small (no more than $M_y=1+2^{-n(\half-\epsilon)}$ for any $\epsilon>0$). Consider using an initial state of
$$
\rho=p\proj{+}^{\otimes n}+(1-p)\frac{\identity}{2^n}
$$
for any $0<p\leq 1$. A single run of the algorithm requires the standard quantum time, $O(2^{n/2})$, and succeeds with probability $\tilde p=p+\frac{1-p}{2^n}$. By allowing $L$ repetitions, there is a time advantage over classical (on average) if $2^{n/2}L<2^{n-1}$, and it succeeds with probability $1-(1-\tilde p)^L\approx\tilde pL$ for small $\tilde p$. As such, for any finite $\epsilon$, $p\sim2^{-n(\half-\epsilon)}$  provides an advantage over classical search algorithms while we evidently anticipate that for smaller $p$, there is less entanglement present.

\subsection{Entanglement}

We turn to calculating the entanglement present in the depolarised algorithm. Progress through the algorithm can be specified by a parameter $\theta$, such that the state is
$$
\rho(\theta)=p\proj{\Psi(\theta)}+(1-p)\frac{\identity}{2^n}
$$
Assume the Schmidt coefficients for the state $\ket{\Psi}$ are $\cos^2\phi$ and $\sin^2\phi$ with respect to bipartition $y$. In the large $n$ limit, $\phi$ can take on any value $(0,\pi/2]$. $\rho(\theta)$ is unitarily equivalent (where the unitaries are local with respect to the bipartition, and hence do not change the eigenvalues under the partial transposition) to a state
\begin{equation}
\left(\begin{array}{ccccc}
p\cos^2\phi+\frac{1-p}{2^n} & 0 & 0 & p\cos\phi\sin\phi & \\
0 & \frac{1-p}{2^n} & 0 & 0 & \\
0 & 0 & \frac{1-p}{2^n} & 0 & \\
p\cos\phi\sin\phi & 0 & 0 & p\sin^2\phi+\frac{1-p}{2^n} & \\
&&&& \frac{1-p}{2^n}\identity_{2^n-4}
\end{array}\right)	\label{eqn:grover}
\end{equation}
where the upper left $4\times 4$ describes the space spanned by $\ket{x_y}\ket{x_{\bar y}}$, $\ket{x_y}\ket{\psi^\perp_{\bar y}}$, $\ket{\psi^\perp_y}\ket{x_{\bar y}}$ and $\ket{\psi^\perp_y}\ket{\psi^\perp_{\bar y}}$. We can readily take the partial transpose of this, and calculate the eigenvalues: $\frac{1-p}{2^n}$ (repeated $2^n-4$ times), $\frac{p}{2}(1\pm\cos(2\phi))+\frac{1-p}{2^n}$,  and $\frac{1-p}{2^n}\pm \frac{p}{2}\sin(2\phi)$. Hence,
$$
M_y=1-\frac{1-p}{2^{n-1}}+p\sin(2\phi).
$$
This is maximised at $\phi=\pi/4$, indicating that for all bipartitions, and at all points during the computation,
$$
M<1+p-\frac{1-p}{2^{n-1}}.
$$
When $p$ is exponentially small, the entanglement is always exponentially small and there is still a provable (oracle-based) speed-up in searching over the classical case. This compares favourably to DQC1 wherein typical unitaries have been numerically shown to produce exponentially small amounts of entanglement \cite{steve}, while having a believed speed-up over classical.

\subsection{Quantum Discord}

Perhaps it is not entanglement that needs to be present in a mixed state quantum computation. Instead, other measures of non-classicality, such as the quantum discord have arisen. If exponentially small entanglement might be considered surprising, perhaps it is the case that there is finite quantum discord? We will now show that the discord is also vanishingly small. A precise definition \cite{discord1,discord2} of the discord is unnecessary. It suffices to know that it is a non-negative quantity which is always 0 in the classical limit and, for pure states, reduces to the entanglement entropy.

We take $\rho(\theta)$ as for the previous calculation,
$$
\left(\begin{array}{ccccc}
p\cos^2\phi+\frac{1-p}{2^n} & 0 & 0 & p\cos\phi\sin\phi & \\
0 & \frac{1-p}{2^n} & 0 & 0 & \\
0 & 0 & \frac{1-p}{2^n} & 0 & \\
p\cos\phi\sin\phi & 0 & 0 & p\sin^2\phi+\frac{1-p}{2^n} & \\
&&&& \frac{1-p}{2^n}\identity_{2^n-4}
\end{array}\right)
$$
and define a second state $\sigma=$
$$
\left(\begin{array}{ccccc}
p\cos^2\phi+\frac{1-p}{2^n} & 0 & 0 & 0 & \\
0 & \frac{1-p}{2^n} & 0 & 0 & \\
0 & 0 & \frac{1-p}{2^n} & 0 & \\
0 & 0 & 0 & p\sin^2\phi+\frac{1-p}{2^n} & \\
&&&& \frac{1-p}{2^n}\identity_{2^n-4}
\end{array}\right).
$$
By virtue of being diagonal with respect to a separable basis, the state $\sigma$ has 0 discord \cite{datta}, and yet is very close in terms of trace distance to $\rho$:
$$
d=\text{Tr}|\rho-\sigma|=p\sin(2\phi).
$$
The quantum discord is continuous \cite{continuousdiscord,cont2}, meaning that we can bound the amount of discord present in $\rho(\theta)$:
$$
D(\rho)\leq 8d(n-1)+4h(d)
$$
where $h(x)=-x\log_2x-(1-x)\log_2(1-x)$ is the binary entropy. This scales as $O(n 2^{-n(\half-\epsilon)})$, and is therefore also vanishingly small across all bipartitions simultaneously, and for all steps of the algorithm.

\subsection{Decomposing the Oracle}

An oracle-based problem (whether this is explicit, as for Grover's search, or implicit, requiring the implementation of controlled-$U$ in DQC1) is useful for determining the computational complexity of a problem (with respect to that oracle), often enabling lower bounds as well as upper bounds. However, it has the potential to mask a lot of the entanglement properties. If one has to implement an oracle by using a collection of smaller, more manageable, operations, we need to be sure that each of those individual operations does not leave an intermediate state that has large amounts of entanglement. After all, there are many quantum computations that start and end in separable states, and it is only the intermediate products that are entangled. The Bernstein-Vazirani algorithm is one case that hides all the entanglement in such a way \cite{BV,BV2}. Typical DQC1 computations probably also hide their entanglement -- imagine decomposing $U$ in terms of one-qubit unitaries and controlled-nots, and implementing the controlled-$U$ by consecutively applying the appropriate controlled-one-qubit and Toffoli gates\footnote{This is certainly far from the only decomposition that one could make, and this is merely an illustrative argument.}. The first Toffoli gate is likely to introduce a finite amount of entanglement.

We claim that there need be no corresponding difficulty when using Grover's search. The Grover iterator may be written as
$$
H^{\otimes n}P_0H^{\otimes n}P_{x_0},
$$
where $H$ is the Hadamard gate, and $P_x$ is an $n$-qubit controlled-phase gate which adds a phase of $\pi$ only to the state $\ket{x}$. Let us take each step in turn. If $\rho_{\text{in}}$ is the state before the action of the iterator, and $\rho_{\text{out}}$ is the state after the application, then our results so far show that both have exponenitally small entanglement and discord. Now, observe that after the action of $P_{x_0}$, the form of $\rho_{\text{in}}$ in Eq.\ (\ref{eqn:grover}) is only changed by adding a negative sign on the off-diagonal elements. Both the entanglement and discord are unchanged. Application of a set of local unitaries ($H^{\otimes N}$) also cannot change these values. What about the application of $P_0$? Instead, we observe that after the action by $P_0$, this is the same as state $\rho_{\text{out}}$ with $H^{\otimes n}$ applied to it, which must have the same entanglement properties as $\rho_{\text{out}}$.

Should we decompose the operations any further, perhaps by writing $P_{x_0}$ in terms of a universal set of two-qubit gates? We suggest this is unnecessary as it appears that in a wide variety of experimental implementations, the gate $P_{x_0}$ can be implemented directly \cite{multi1,multi2,multi3}. So, we are justified in our claim that every elementary step of the algorithm leaves us in a state of almost no entanglement and almost no discord.

\section{Conclusions}

One of the aspects of the DQC1 model of computation that originally generated much interest was the suggestion that it achieved its computational speed-up without entanglement. In fact, it does use entanglement \cite{steve}, and that entanglement can be up to a finite amount across many different bipartitions of the system. Even on those bipartitions for which there is no entanglement, it has been shown \cite{discord} instead that other measures of non-classicality are present in finite amounts and could be a resource for the computational speed-up that is believed to be present in DQC1.

In this paper, we have contrasted this with Grover's search when acting on an initial state that has a large admixture of the maximally mixed state. The degree of admixture for which a computational speed-up is still possible (without any attempt at error correction) was readily derived and, as such, it was shown that over every bipartition both the entanglement and quantum discord are vanishingly small. Moreover, we know that the computational speed-up is present for Grover's search, since we know the minimum running time for the classical algorithm (using a certain oracle). We conclude that DQC1 is nothing special -- Grover's search exhibits stronger properties in every way (guaranteed c.f.\ assumed speed-up, guaranteed exponentially little entanglement c.f.\ little entanglement on average, exponentially little discord c.f.\ finite discord), except for one feature -- DQC1 is believed to have an exponential speed-up while Grover's search is only polynomial.

The presence of quantum discord in the DQC1 model has been taken by a number of authors as ``the first real evidence that mixed-state quantum computation can have an advantage over classical computation even when entanglement is absent'' \cite{revmod}, choosing to view the almost-no entanglement as essentially equivalent to no entanglement. However, here, we have almost no entanglement, and almost no quantum discord (and, indeed, any measure of quantum resource that is continuous). We are left to conclude that there is a marked difference between `almost no' and `no' entanglement/discord. Exactly what is required for a computation to gain speed over a classical one remains unclear, but it seems that only small amounts of whatever resource suffice. However, this leaves open the possibility that larger amounts of a resource, such as quantum discord, are necessary for an exponential  vs.\ polynomial speed-up.

On a final note, we wish to raise an important issue with the calculation of the typical entanglement \cite{steve} or discord \cite{discord} across a bipartition in the DQC1 computation. The decision problem for DQC1 was defined by explicitly bounding the trace away from 0. However, typical unitaries have an expected trace of 0, and the trace is closely centred on the 0 value. As such, the decision problem explicitly asserts that the unitaries involved are not typical, and further studies are required to clarify the impact of this, although the studies of special cases that are performed in the Appendix suggest that the impact is not significant.

{\em Acknowledgements:} We thank L.C.\ Kwek for introductory conversations.

%\clearpage
\newpage
\appendix

\section{Typical Entanglement}

In this Appendix, we illustrate some important features of the calculation of typical entanglement in the DQC1 model. Rather than permitting an arbitrary Haar-random unitary, we concentrate on a subset of unitaries, based on graph states, that facilitate simpler calculations. We define the random unitary, $U$, in the following way: consider $n$ qubits and 
select a random diagonal matrix, meaning that each diagonal element is 
selected at random to be $e^{i\theta_x}$ for all $x\in\{0,1\}^n$. We then conjugate this with Hadamard gates on each 
qubit (i.e.\ both before and after the diagonal matrix) before conjugating by 
controlled phase gates between all pairs of qubits $(i,i+1)$. These conjugations are the unitaries required for preparing the one-dimensional cluster state.

The reason for concentrating on graph states is that there is a well established formalism for calculating the partial transpose \cite{kay1,kay2,kay3}. If the eigenvalues of $U$ are written as a vector
$$
\ket{\Theta}=\sum_xe^{i\theta_x}\ket{x},
$$
then the eigenvalues of $U_{01010101\ldots 01}$ are given by the vector $R\ket{\Theta}$ where
$$
R=H^{\otimes n}\left(\sum_{x\in\{0,1\}^n}(-1)^{\sum_{i=1}^{n-1}x_ix_{i+1}}\proj{x}\right)H^{\otimes n}.
$$
%Recall that what we want to do is find the sum of the moduli of the entries of $R\ket{\Theta}$ for those entries that have an absolute value greater than 1.

\begin{lemma}
If $n$ is even, then every matrix element of $R$ has $|\bra{y}R\ket{z}|=\frac{1}{2^{n/2}}$. In every row, there are $\half(2^n+2^{n/2})$ positive entries, and $\half(2^n-2^{n/2})$ negative entries.
\end{lemma}
\begin{proof}
The matrix elements have the form
$$
\bra{y}R\ket{z}=\frac{1}{2^n}\sum_{x\in\{0,1\}^n}(-1)^{x\cdot(z+y)}(-1)^{\sum_{i=1}^{n-1}x_ix_{i+1}}.
$$
First, observe that $\bra{y}R\ket{z}=\bra{000\ldots 0}R\ket{z\oplus y}$, so it suffices to concentrate solely on the first row. The elements of all other rows are just permutations of the first. We will denote the elements of this first row by $R_z$.

Notice that
$$
\sum_{z\in\{0,1\}^n}R_z=1.
$$
This means that if it is true that all $2^n$ elements have $|R_z|=1/2^{n/2}$, then since they are all real, it must be that there are $\half(2^n+2^{n/2})$ positive entries, and $\half(2^n-2^{n/2})$ negative entries.

We will prove the first half of the lemma by induction, using the base case of $n=2$:
$$
R(00)=\half\qquad R(01)=\half\qquad R(10)=-\half\qquad R(11)=\half.
$$

Assume that $|R^{(n-2)}(z)|=\frac{1}{2^{n/2-1}}$. Now consider $R^{(n)}(z\|z_{n-1}z_{n})$. For each term $x$ in the sum for $R^{(n-2)}(z)$, there are 4 terms in $R^{(n)}(z\|z_{n-1}z_{n})$: $x\|x_{n-1}x_n$. We can consider each of them separately, and how they affect the additional phases
$$
(-1)^{z_nx_n+z_{n-1}x_{n-1}+x_nx_{n-1}+x_{n-1}x_{n-2}}.
$$
One of the phases can be moved into an effective $z$:
\begin{equation*}
\begin{split}
R^{(n)}(z\|z_{n-1}z_{n})=&\frac{1}{4}\sum_{x_n,x_{n-1}}(-1)^{z_nx_n+z_{n-1}x_{n-1}+x_nx_{n-1}}\times	\\
&R^{(n-2)}(z\oplus(00\ldots 0x_{n-1})).
\end{split}
\end{equation*}
On performing the sum over $x_n$, we either have $\half R^{(n-2)}(z)$ (if $z_n=0$) or $\half R^{(n-2)}(z\oplus(00\ldots 01))$. By assumption both have the same absolute value, $1/2^{n/2}$.
\end{proof}

\begin{figure}[!tb]
\begin{center}
\includegraphics[width=0.45\textwidth]{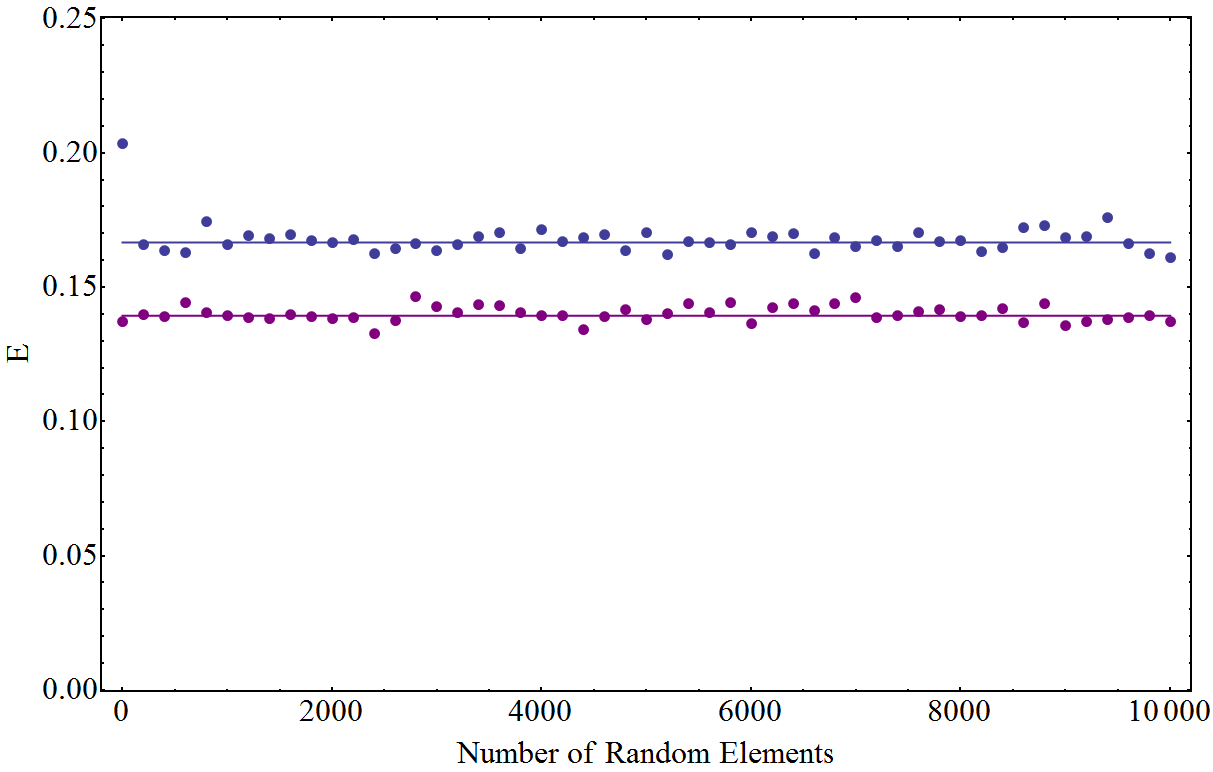}
\end{center}
\vspace{-0.5cm}
\caption{Comparison of theoretical results (lines), Eq.\ (\ref{eqn:E}), and average for 10000 samples (points) for both 1D (upper, blue) and 2D (lower, purple) random walks.}\label{fig:random}
\vspace{-0.4cm}
\end{figure}

Our aim is to evaluate the multiplicative negativity:
$$
M_y=1+\frac{1}{2^n}\sum_{\lambda\in R\ket{\Theta}:\lambda>1}(\lambda-1).
$$
We assume that each element of $R\ket{\Theta}$ is independent  (this may not be strictly true, but is a good approximation). This reduces $M_y$ to $M_y=1+E$ where $E$ is the expected value of $\max(\lambda-1,0)$. To evaluate this, we need the probability distribution for $\lambda$. As a consequence of the Lemma, we have
$$
P(\lambda\geq1)=P\left(\left|\sum_{i=1}^{\half(2^n+2^{n/2})}\!\!\!e^{i\theta_i}-\sum_{i=\half(2^n+2^{n/2})+1}^{2^n}\!\!\!\!\!\!\!e^{i\theta_i}\right|\geq 2^{n/2}\right).
$$
However, if the probability distribution of $\theta_i$ is unaffected by a $\pi$ shift, this is entirely equivalent to
$$
P\left(\left|\sum_{i=1}^{2^n}e^{i\theta_i}\right|\geq 2^{n/2}\right).
$$

\begin{figure}[!t]
\begin{center}
\includegraphics[width=0.45\textwidth]{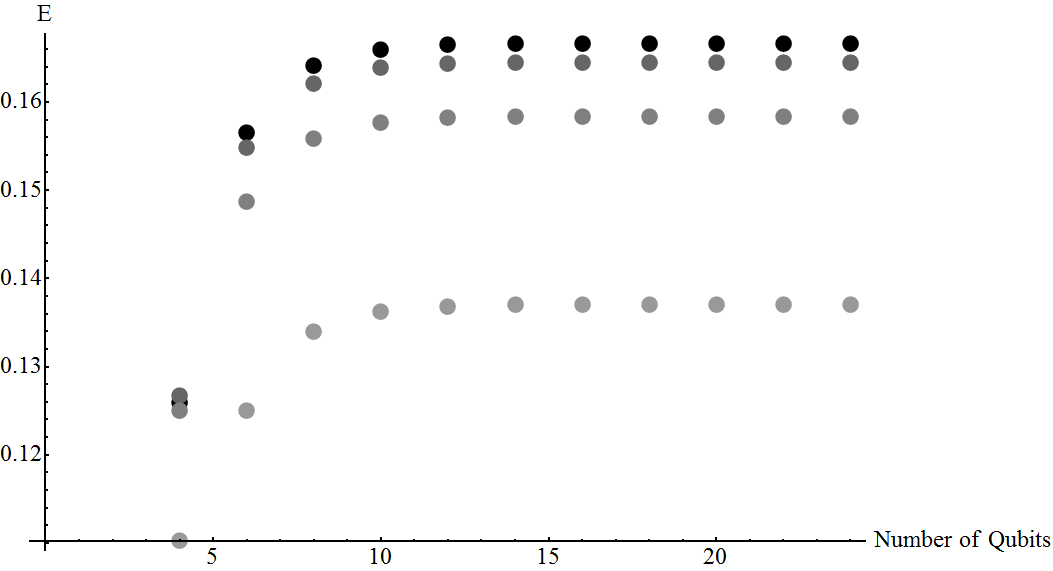}
\end{center}
\vspace{-0.5cm}
\caption{Typical entanglement for a graph diagonal unitary of fixed trace $\Tr(U)/2^n=\frac{1}{4},\frac{5}{8},\frac{3}{4},\frac{7}{8}$ for shades black through to light grey respectively. The larger the trace, the less entanglement is present.}\label{fig:random2}
\vspace{-0.5cm}
\end{figure}

We consider two cases. Firstly, each $\theta_i$ is a random choice, $\theta_i\in\{0,\pi\}$. In this case, 
$
\left|\sum_{i=1}^{2^n}e^{i\theta_i}\right|
$
is simply the expected distance of a random walk in 1D. Secondly, $\theta_i\in[0,2\pi)$ leads to the interpretation of
$
\left|\sum_{i=1}^{2^n}e^{i\theta_i}\right|
$
as a random walk in 2D\footnote{We define this to be the walk such that at each step, a step length of 1 is taken in a random direction in the plane. Some authors choose to define it as a walk on a square lattice.}. In either case, Rayleigh's solution for the probability distribution at large $n$ is
$$
P\left(x\leq \left|\sum_{i=1}^{2^n}e^{i\theta_i}\right|<x+\delta x\right)=\left\{\begin{array}{cc}
\sqrt{\frac{2}{\pi 2^n}}e^{-x^2/2^{n+1}}dx & \text{1D} \\
\frac{2x}{2^n}e^{-x^2/2^n}dx & \text{2D} 
\end{array}\right.
$$
in order to calculate
$$
E=\int_{2^{n/2}}^\infty P\left(x\leq \left|\sum_{i=1}^{2^n}e^{i\theta_i}\right|<x+\delta x\right)\left(\frac{x}{2^{n/2}}-1\right),
$$
which yields
\begin{equation}
E=\left\{\begin{array}{cc}
\sqrt{\frac{2}{\pi e}}-\text{erfc}\left(\frac{1}{\sqrt{2}}\right) & \text{1D} \\
\frac{\sqrt{\pi}}{2}\text{erfc}(1) & \text{2D}
\end{array}\right.
\label{eqn:E}
\end{equation}
Fig. \ref{fig:random} provides numerical confirmation of this calculation. We conclude that in both cases there is a constant amount of entanglement. This is in contrast with the numerical results in \cite{steve} for general random unitaries.

However, there is one further subtlety that it is important to raise. Recall that for the DQC1 model we are given the promise that the trace of $U$ is bounded away from 0. Meanwhile, the probability distribution for the trace of a typical unitary is strongly centred on the value 0 -- in effect, we are post-selecting on highly atypical unitaries so typicality arguments might not reveal everything. However, it turns out that this only serves to reduce the amount of entanglement present. For a simple argument, consider the case where $\Tr(U)=2^n$. In this case, we know that $U=\identity$, and that $M_y=1$, less than the above-calculated values. More generally, we have numerically examined the random unitary model described above (restricting the eigenvalues to $\pm 1$), fixing the trace to be different values of $\Tr(U)$, see Fig.\ \ref{fig:random2}. The amount of entanglement present decreases as the value of the trace is increased.

% Indeed, one can make a similar argument to that of Theorem \ref{thm:1}: assume that $U$ only has eigenvalues $\pm 1$, and that it has trace $2k-2^n$. This means it has $k$ eigenvalues $+1$ and $2^n-k$ eigenvalues $-k$. We can write
% $$
% U=\identity-2R
% $$
% where $R=\sum_{i=1}^{2^n-k}\proj{\phi_i}$ for a set of orthogonal states $\ket{\phi_i}$. Similarly, $U_y=\identity-2R_y$, and $\eta_i=|1-2\gamma_i|$ where $\eta_i$ are the singular values of $U_y$ and $\gamma_i$ are the singular values of $R_y$.

\end{document}